\documentclass[a4paper,12pt]{article}
\usepackage{amsmath,amssymb,amsthm,mathtools}

\usepackage[colorlinks,citecolor=blue,urlcolor=magenta]{hyperref}
\usepackage{doi}
\usepackage{booktabs}

\usepackage[style=
authoryear-comp, 
sorting=nyt, 
dashed=false, 
maxcitenames=2, 
maxbibnames=99, 
uniquelist=false,
uniquename=false,
giveninits=true, 
natbib, 
date=year 
]{biblatex}

\AtBeginRefsection{\GenRefcontextData{sorting=ynt}}
\AtEveryCite{\localrefcontext[sorting=ynt]}

\DeclareFieldFormat{pages}{#1} 
\renewbibmacro{in:}{\ifentrytype{article}{}{\printtext{\bibstring{in}\intitlepunct}}} 

\usepackage[inline,shortlabels]{enumitem}
\setlist[enumerate,1]{label=(\roman*)}

\usepackage[margin = 1.25in]{geometry}
\usepackage{setspace}
\title{Recent Advances on Uniqueness of Competitive Equilibrium\thanks{This article was written for the 50th anniversary issue of the \textit{Journal of Mathematical Economics}. We thank {\"O}mer A\c{c}{\i}kg{\"o}z, Costas Arkolakis, Bar Light, and Eric Young for helpful comments.}
}
\author{Alexis Akira Toda\thanks{University of California San Diego. Email: \href{mailto:atoda@ucsd.edu}{atoda@ucsd.edu}.} \and Kieran James Walsh\thanks{ETH Z\"{u}rich. Email: \href{mailto:kwalsh@ethz.ch}{kwalsh@ethz.ch}.}}

\date{\today}

\newcommand{\cP}{\mathcal{P}}
\newcommand{\ES}{\mathrm{ES}}
\newcommand{\RRA}{\mathrm{RRA}}
\newcommand{\abs}[1]{\left\lvert #1 \right\rvert}
\newcommand{\norm}[1]{\left\lVert #1 \right\rVert}

\newcommand{\R}{\mathbb{R}}
\def\maximize{\operatornamewithlimits{maximize}}
\def\st{\operatorname{subject~to}}
\def\ie{\textrm{i.e.}}
\def\eg{\textrm{e.g.}}
\newcommand{\set}[1]{\left\{ #1 \right\}}
\def\E{\operatorname{E}}
\def\e{\mathrm{e}}

\theoremstyle{plain}
\newtheorem{thm}{\bf Theorem}
\newtheorem{lem}[thm]{\bf Lemma}
\newtheorem{prop}[thm]{\bf Proposition}

\numberwithin{equation}{section}
\numberwithin{thm}{section}

\begin{document}
\maketitle

\begin{abstract}

This article reviews the recent advances in the uniqueness and multiplicity of competitive equilibria in models arising in mathematical economics, finance, macroeconomics, and trade.

\medskip


\medskip

\end{abstract}

\section{Introduction}\label{sec:intro}

For most curious minds, understanding the classic supply and demand diagram for the first time results in epiphany. Simple and plausible assumptions about the desires of consumers and costs of producers lead to stark predictions about prices and economic activity in markets that at first seem chaotic and complicated. For example, a reduction in costs pushes down prices and increases output. But the clarity the classic model provides stems from the center of the diagram’s famous $\times$, that is, the unique point at which the price clears the market, supply equals demand, and deviations from the point sensibly converge back to it. If, say, the demand curve is non-monotonic in prices, then the $\times$ can become a rod of Asclepius, with demand snaking around supply and yielding many equilibrium points.\footnote{See, for example, \citet{Dougan1982}.} And with many equilibria, the predictions of equilibrium are ambiguous: the equilibria disagree on prices and quantities, and there can be unstable equilibrium points that defy the normal logic of convergence to them.

Therefore, to the extent that the goal of economics is to predict the result of market interactions and perform comparative statics and counterfactuals, the issue of uniqueness should be central to economics. But, perhaps unfortunately, it is easy to construct simple examples with multiple equilibria.\footnote{For instance, \citet{TodaWalsh2017ETB} construct multiple equilibria in Edgeworth box economies with CRRA, quadratic, quasi-linear, and general additively separable preferences.} On the other hand, there are numerous theories across many fields that admit a unique and stable equilibrium. This tension was described by \citet{Kehoe1984} in the \textit{Journal of Mathematical Economics}: 
\begin{quote}

``There are, of course, possibilities that are always present as long as we are not sure that a specific model has a unique equilibrium. Conditions that ensure uniqueness, however, seem to be too restrictive to have general applicability [\ldots]. Applied economists seem to view non-uniqueness of equilibrium as pathological. Theorists, on the other hand, seem to have accepted it as commonplace.''

\end{quote}

In this article, we survey recent advances on the theory of uniqueness of competitive equilibrium, highlighting contributions from the \textit{Journal of Mathematical Economics}. While our focus is on newer papers in the field of general equilibrium theory, we also discuss the classics as well as more applied theory in macroeconomics, trade, and finance.\footnote{Since all of the results we discuss have the common goal of establishing uniqueness of competitive equilibrium, our field labeling is somewhat arbitrary, and for our purposes the papers differ only in style and assumptions. Nevertheless, we hope our organization by field illuminates the history of thought and guides readers interested in particular applications.} 

To our knowledge, there has not been a survey of this kind in many years. \citet{Mas-Colell1991}, \citet{Kehoe1998}, and \citet{HensPilgrim2002} provided insightful summaries on which we draw. Other excellent surveys include Chapter 9 of \citet{ArrowHahn1991} and Chapter 17 of \citet{MWG1995}. But there have been numerous advances in the last 20 years, which we highlight in this article. Table \ref{tab:papers} lists the uniqueness conditions from many of the papers we discuss.

One caveat before we proceed: several of the papers we discuss are not explicitly about uniqueness. We attempt to provide context when possible, but note that in many cases the uniqueness results were not the main points of the papers.

\section{Anything Goes}\label{sec:anything}

The difficulty of establishing uniqueness can be illustrated with a simple example that would be familiar to any undergraduate economics major. Consider an economy with two agents ($1$ and $2$) and two goods (with consumption denoted by $x_1$ and $x_2$). Each agent is endowed with $e>0$ units of each good. Normalize to $1$ the price of good $1$, and let $p$ be the price of good $2$. The utility functions of the two agents are
\begin{align*}
    U_1(x_1,x_2) & = \alpha u(x_1) + (1-\alpha) u(x_2), \\
    U_2(x_1,x_2) & = (1-\alpha) u(x_1) + \alpha u(x_2),
\end{align*}
where $u$ is the common Bernoulli utility, and the utility weight $0<\alpha<1$ is symmetric across the agents. If $u(x)=\log x$, the agents' demands for good $1$ are $\alpha e (1+p)$ and $(1-\alpha) e (1+p)$. Since the aggregate supply of good $1$ is $2e$ and the aggregate demand $e(1+p)$ is linear in $p$, the unique equilibrium occurs at price $p=1$, as one would likely expect due to the symmetry and simplicity of the economy.

But now suppose the common Bernoulli utility is quadratic, $u(x) = x - \frac{1}{2\tau}x^2$, where $\tau>0$ governs the curvature (and risk tolerance in the case of uncertainty), and assume $\tau > e$ (which rules out consumption above the satiation point). As with log utility, symmetry yields an equilibrium at $p=1$. But, as \citet[Proposition 2]{TodaWalsh2017ETB} show, if $\alpha < \frac{2-\sqrt{2}}{4}$ or $\alpha > \frac{2+\sqrt{2}}{4}$, $\alpha=1/10$ for example, then there are exactly two more equilibria (for a total of three).\footnote{The reason is that low or high $\alpha$ renders total demand for good $1$ strictly increasing in its price ($1/p$) at $p=1$. Since total demand is continuous and must be above supply for low prices, there must be at least two more equilibria (and in this example exactly two).} 

What happened? The economy has strictly concave utilities, only two goods, and identical endowments, and the only source of heterogeneity comes from utility weights that are symmetric across the agents. Nonetheless the economy has three equilibria. Moreover, by switching to a different concave utility (log for example), we can restore uniqueness.

This example shows that, on the surface at least, there is a fine line between equilibrium models providing stark predictions or implying ``anything goes.'' Fortunately, as we describe below, quadratic utility's unattractive property of increasing absolute risk aversion can help drive multiplicity of equilibria. So, echoing \citet{Kehoe1984}, while the simplicity of this example makes non-uniqueness seem commonplace, quadratic utility is arguably pathological.

The term ``anything goes'' is often used as shorthand for the famous \citet{Sonnenschein1972,Sonnenschein1973}-\citet{Mantel1974}-\citet{Debreu1974} (SMD) theorems, the latter of which appeared in the very first 1974 issue of the \textit{Journal of Mathematical Economics}. To illustrate these results, suppose there are $J$ goods and agent $i \in I\coloneqq \set{1,\dots,I}$ solves
\begin{subequations}\label{eq:opt}
\begin{align}
    &\maximize && U_i(x) \\ 
    &\st && p \cdot x = p \cdot e_i,
\end{align}
\end{subequations}
where $U_i:\R_{++}^J\to \R$ is agent $i$'s (continuous and strongly monotonic) utility function, $p \gg 0$ is the vector of prices, and $e_i \in \R^J_+$ is $i$'s vector of endowments. Define $w_i = p \cdot e_i$ to be $i$'s income. Let demand $x_i(p,w_i)$ be the solution to Problem \ref{eq:opt}. Define $z_i(p) \coloneqq x_i(p,w_i) - e_i$ to be individual excess demand, let $z(p) \coloneqq \sum_i z_i(p) $ be the market excess demand function, and let $e = \sum_i e_i$ be the total endowment. \emph{Competitive equilibrium} consists of
\begin{enumerate*}
    \item prices $p^* \gg 0$ satisfying $z(p^*) = 0$ and
    \item the corresponding consumption allocations $x_i^* = x_i(p^*,p^* \cdot e_i)$.
\end{enumerate*}

Since the utilities are strictly increasing, the agents completely spend their endowments, implying $p \cdot z_i(p) = 0$. Adding across the agents, we see that market excess demand satisfies \emph{Walras' Law}: $p \cdot z(p)=0$ (the value of demand equals the value of the endowment for any prices). Therefore, one of the market clearing conditions always follows from the others. Furthermore, the individual excess demands $z_i(p)$ are homogeneous degree zero because only relative prices matter for solving Problem \ref{eq:opt}, implying $z(p)$ is homogeneous degree zero. So without some normalization, we can always form many equilibria by rescaling the equilibrium price vector. But this form of multiplicity is uninteresting because it does not affect consumption allocations, so we restrict possible equilibrium price vectors to $p \in \cP\coloneqq \set{p \in \R^J_{++}: p_1 = 1}$ (the first good is the num\'eraire).\footnote{Many papers use alternative normalizations, \eg $\sum_j p_j = 1$.}

The question of the uniqueness literature is, under what conditions is there a unique $p \in \cP$ satisfying $z(p) = 0$? To answer this question, we must derive the properties of $z$ from the underlying assumptions about the economy regarding preferences and endowments. In our example, we showed that market excess demand $z$ is homogeneous degree zero and satisfies Walras' Law. What else can we say about $z$? It turns out the answer, even in our restrictive example, is: not very much.

If we assume each $U_i$ is differentiably strictly quasi-concave (see \citet[p.~14]{VillanacciCarosiBenevieriBattinelli2002} for the definition), an application of the implicit function theorem implies that individual excess demands are continuously differentiable, implying market excess demand is also continuously differentiable. Consequently, as \citet{Debreu1970} shows, the economy is \emph{regular} for almost all initial endowments. Loosely, regular economies are ones in which market excess demand has non-zero slope at equilibrium prices, implying a finite number of equilibria. 

We also know that our individual excess demands satisfy the \emph{weak axiom of revealed preferences} (WARP; \ie $q \cdot z_i(p) \le 0$ and $p \cdot z_i(q) \le 0$ imply $z_i(p) = z_i(q)$).\footnote{In other words, WARP says that if an agent's demand at $p$ is affordable when facing prices $q$ but different from demand at $q$, demand at $q$ must be unaffordable when facing prices $p$ (as demand at $q$ is ``revealed preferred'' to demand at $p$). In our example \eqref{eq:opt}, the problem has a unique solution for any price vector, so if demand at $p$ is affordable at prices $q$, demand at $q$ provides strictly more utility (and must thus be unaffordable at prices $p$).} If market excess demand were to inherit WARP, equilibrium would be unique for almost all endowments: WARP would imply the set of equilibrium prices is convex,\footnote{Following \citet{Kehoe1998}, suppose $p$ and $q$ are equilibrium price vectors, and let $p(\theta) = \theta p + (1-\theta) q$, $\theta \in [0,1]$.  If $z(p(\theta)) \neq 0$, by Walras' Law, $(\theta p + (1-\theta) q) \cdot z(p(\theta)) = 0$, so we must have either $p \cdot z(p(\theta)) \le 0$ or $q \cdot z(p(\theta)) \le 0$. But since $z$ satisfies WARP, we must then have $p(\theta) \cdot z(p) > 0$ or $p(\theta) \cdot z(q) > 0$, either of which contradicts $z(p)=z(q)=0$. Thus, $z(p(\theta)) = 0$, proving the set of equilibrium price vectors is convex under WARP.} but by \citet{Debreu1970} the set is generically discrete, so the equilibrium set would be a single vector (generically). But, unfortunately, even in our example economy market excess demand does not necessarily inherit WARP. And this is the heart of the anything goes problem. Without WARP, an economy cannot behave as if there were an standard representative agent because the demands of any representative agent with standard preferences would satisfy WARP.

The SMD theorems tell us that we can find individual preferences and endowments that generate arbitrary continuous market excess demand functions, including ones that admit an arbitrary finite number of equilibria and thus violate WARP. Consider \citet{Debreu1974}, for example. Letting
\begin{equation*}
    S\coloneqq \set{p \in \R^J : p \gg 0,\norm{p} =1}
\end{equation*}
(where $\norm{p}\coloneqq \sum_j \abs{p_j}$ denotes the $\ell^1$ norm) and $S_\varepsilon\coloneqq \set{p \in S : (\forall j) p_j \ge \varepsilon}$, he shows that any continuous function $f: S \to \R^J$ satisfying Walras' Law on $S$ can (for every $\varepsilon>0$) be constructed on $S_\varepsilon$ as the sum of the excess demands of $J$ agents with strictly-convex, monotone, continuous, complete preferences on $\R^J_+$ and endowments in $\R^J_+$.\footnote{\citet{Geanakoplos1984}, in the \textit{Journal of Mathematical Economics}, shows how to construct \citet{Debreu1974}'s underlying utility functions.} Therefore, an endowment economy populated with agents with standard preferences can aggregate into an exotic, non-differentiable market excess demand function with many equilibrium points. Indeed, virtually anything is possible.

At the same time, it was known in the 1970s that homotheticity assumptions could imply a representative agent and uniqueness (\citet{Gorman1953} aggregation). In our example \eqref{eq:opt} , suppose for all $i\in I$ that
\begin{equation*}
    U_i(x)=U(x)\coloneqq \sum_{j=1}^J \beta_j \frac{x_j^{1-\gamma}}{1-\gamma},
\end{equation*}
where $\beta_j>0$ and $\gamma>0$ is the relative risk aversion coefficient. (The case $\gamma=1$ corresponds to log utility.) With these identical, constant relative risk aversion (CRRA) preferences, agent optimality implies
\begin{equation*}
    \frac{\beta_j}{\beta_1}(x_{ij}/x_{i1})^{-\gamma}=p_j\iff \frac{x_{ij}}{x_{i1}}=\left(\frac{\beta_1}{\beta_j}p_j\right)^{-1/\gamma}
\end{equation*}
for any $j \in J$. Summing over $i \in I$,\footnote{If $\frac{X_1}{Y_1}=\frac{X_2}{Y_2}=Z$ then $\frac{X_1+X_2}{Y_1+Y_2}=Z$.} using market clearing, and rearranging, we arrive at $\frac{\beta_j}{\beta_1} \left( \frac{e_j}{e_1} \right)^{-\gamma} = p_j$ (where $e_j$ and $e_1$ refer to total endowments). Therefore, equilibrium prices are uniquely given by the preference parameters and aggregate endowment, and it is as if there were a representative agent with the common preferences and endowed with $e$. Intuitively, when preferences scale-up proportionally, heterogeneity can ``average out.'' One might thus hope that imposing homothetic preferences can let us say much more about $z$. However, the remarkable result of \citet{Mantel1976} dispels this hope: he proves a version of SMD with homothetic agent preferences.\footnote{\citet{Hens2001}, in the \textit{Journal of Mathematical Economics}, extends \citet{Mantel1976} to incomplete markets.} \citet[Proposition 1]{TodaWalsh2017ETB} show how to construct two-good, two-agent, symmetric CRRA versions of example \eqref{eq:opt} with three equilibria (an thus market excess demand violating WARP).

The difficulty is that even with CRRA utility, individual demands can display Giffen behavior due to income effects. For example, \citet{Won2023} shows that with two goods and CRRA utility, demand always becomes strictly upward sloping above some price threshold due to income effects. With a sufficiently high price, any agent will become a seller, and once the price rises far enough the income effect from selling will make demand upward sloping. So even with CRRA utility, without some restriction on endowments and utility weights (the $\beta_j$'s above, which could be individual-specific), it is possible to find an equilibrium where market demand displays Giffen behavior. But since market demand rises above supply for very low prices, and since demand is continuous in this context, a Giffen equilibrium in the CRRA economy implies non-uniqueness: demand will be below supply for a slightly lower price, implying another equilibrium as demand eventually rises above supply at very low prices.

Income effects are the heart of the difficulty in establishing uniqueness, since compensated demand (the substitution effect) is downward sloping under general conditions. Even with relatively strong assumptions on utility, individual excess demands can be non-monotonic since price changes affect wealth. And if, due to heterogeneity, individual demand inflection points occur across a large price range, market demand can take many shapes and intersect supply at many prices.

The discussion of this section illustrates the tight link between the questions of aggregation and uniqueness.\footnote{See \citet{Hildenbrand1999} and the corresponding 1999 special issue of the \emph{Journal of Mathematical Economics} on demand aggregation.} If a heterogeneous agent economy is equivalent to an economy with a representative agent with standard preferences, then equilibrium will be unique because there is no scope for Giffen behavior at the equilibrium. When you cannot construct a standard representative agent, uniqueness is still possible, but the problem becomes more difficult because you cannot appeal to properties of a single agent's demand function. The \citet{Gorman1953} representative agent is the most famous example of aggregation, but the proof we provided above (adding across Euler equations with common utility weights) also works in two other cases, heterogeneous quadratic Bernoulli utility\footnote{The authors of this survey learned the quadratic result in John Geanakoplos' 2008 graduate micro theory course. The result appears in \citet{HensPilgrim2002} but without citation, and it seems to have been known for a long time.} and heterogeneous constant absolute risk aversion (CARA) Bernoulli utility.

But there are at least three sets of assumptions not covered by the SMD theorems: quantitative bounds on utilities, restrictions on the joint distribution of preferences and endowments, and restrictions on the number of goods or agents. The remainder of this article describes such assumptions, many of which commonly appear in applied work, that are sufficient for uniqueness.

\section{General Equilibrium Theory}

\subsection{Classic Uniqueness Results}\label{sec:classic}

As described in Section \ref{sec:anything}, equilibrium is unique if the economy is regular and market excess demand satisfies WARP. This result, attributed to \citet{Wald1951} by \citet{ArrowHahn1991} and
 \citet{Kehoe1998}, holds because the set of equilibrium prices is both convex (by WARP) and discrete (by regularity). Another classic sufficient condition for uniqueness is the ``gross substitutes'' property, which \citet{ArrowHahn1991} also attribute to \citet{Wald1951}. As described in \citet{MWG1995}, market excess demand $z$ has the gross substitute property if $z_j(p') > z_j(p)$ whenever $p'_k>p_k$ for some $k$ and $p'_j=p_j$ for $j \neq k$. That is, demand rises for a good when the price of a different good increases. In an endowment economy (without production), the gross substitutes property is sufficient for uniqueness (see \citet{MWG1995} for an intuitive proof).

But WARP and gross substitutes are conditions on market demand, not on the economy's primitives, and the key question in applied settings is how the underlying preferences and endowments map into the number equilibria. In Section \ref{sec:anything} we discussed that identical homothetic utility yields uniqueness (\citet{Gorman1953} aggregation), but there are several other similarly concise results long-known in the literature. Consider an $I$-agent, $J$-good economy with additively separable utility functions $U_i(x)=\sum_j u_{ij}(x)$, where $u_{ij}'>0$ and $u_{ij}''<0$. If the relative risk aversion for every agent $i$ and good $j$, $-xu_{ij}''(x)/u_{ij}'(x)$, is everywhere bounded above by one, then individual as well as market excess demands satisfy the gross substitutes property. \citet{HensLoeffler1995} prove this, citing related results in \citet{Varian1975}, \citet{Wald1951}, \citet{Slutsky1915}, and (for an infinite number of goods) \citet{Dana1993}.\footnote{See also \citet{Dana1993jme} in the \emph{Journal of Mathematical Economics}.} \citet{HensLoeffler1995} also prove uniqueness for two other stark cases:
 \begin{enumerate*}
     \item CARA Bernoulli utility and
     \item constant aggregate endowment, agent-specific Bernoulli utility, and common utility weights.
 \end{enumerate*}
In our example \eqref{eq:opt}, the first case means
\begin{equation*}
    u_{ij}(x)=-\frac{\beta_{ij}}{\alpha_i}\e^{-\alpha_i x_j},
\end{equation*}
allowing for heterogeneous risk aversion $\alpha_i>0$ and utility weights $\beta_{ij}>0$. The second case means $u_{ij}(x)=\beta_ju_i(x)$ for individual-specific $u_i$ and good-specific $\beta_j$, and $e_j = e_k$ for all $j,k \in J$ (the goods all have the same supply).\footnote{When we interpret the weights as probabilities and the goods as state-dependent consumption, common weights mean homogeneous beliefs, and constant aggregate endowments means there is no aggregate risk.}

 However, neither relative risk aversion less than one, nor identical homothetic utility, nor CARA, nor constant aggregate endowments is necessary for uniqueness. The literature discovered that collinear endowments can allow for weaker assumptions on utility while still preserving uniqueness. In our example \eqref{eq:opt}, endowment collinearity means $e_i = \omega_i e$, for the set of constants $\omega_i>0$. This case allows for wealth inequality but implies that the wealth distribution is independent of the price vector since $(p \cdot \omega_i e)/(p \cdot e) = \omega_i$. \citet{Eisenberg1961} and \citet{Chipman1974} show that with heterogeneous homothetic utility and collinear endowments, you get a representative agent with a utility function equal to the sum of $\log$ individual utilities weighted by $\omega_i>0$, maximized subject to the aggregate resource constraint (see \citet{Kehoe1998}). The representative agent problem satisfies standard assumptions, so WARP holds and equilibrium is unique. \citet{MitjushinPolterovich1978} show that with collinear endowments you can drop hometheticity, provided that RRA is everywhere less than four.\footnote{\citet{Bettzuege1998}, in the \emph{Journal of Mathematical Economics}, extends \citet{MitjushinPolterovich1978} to incomplete markets. \citet{Dana1995}, also in the \emph{Journal of Mathematical Economics}, extends the results to an infinite number of goods and incomplete markets.} 

Both \citet{MitjushinPolterovich1978} and \citet{HensLoeffler1995} (their Equation (7), for example) illustrate the importance of utility curvature for uniqueness: with a high elasticity of substitution (low risk aversion), income effects are dampened, allowing more scope for uniqueness. On the other hand, when the elasticity of substitution is low and agents have a high desire for smoothness, income effects are amplified, and there is more scope for Giffen behavior.

\subsection{Recent Advances}

Several recent papers have made progress on the uniqueness problem by assuming there are only two goods. This restrictive approach is useful for two reasons. First, two period models (in which the goods could be consumption today and tomorrow and the price is the interest rate) are widely employed in applied settings. Second, assuming two goods greatly simplifies analysis because it implies market excess demand is a univariate function. 

Suppose in our example \eqref{eq:opt} that there are two goods, $J=2$. \citet{GeanakoplosWalsh2018Uniqueness} further assume that all agents have the same endowment $e_i = (e_1,e_2)$, as well as the same Bernoulli utility. Heterogeneity stems purely from the utility weights, which we can think of as representing time preference if the two goods are consumption today and tomorrow. Despite the many assumptions in this setting, it is relevant in that versions of it arise in \citet{DiamondDybvig1983}-type models, in which ex ante identical agents face idiosyncratic time preference shocks after some initial portfolio choice decision. \citet{GeanakoplosWalsh2018Uniqueness} show that decreasing (nonincreasing) absolute risk aversion (DARA) implies uniqueness.\footnote{\citet{GeanakoplosWalsh2018Liquidity} use this result for signing welfare comparative statics in \citet{DiamondDybvig1983}-type economies.} In the notation of our example \eqref{eq:opt}, DARA means
\begin{equation*}
    U_i(x)=\sum_j \beta_{ij}u_i(x_j),
\end{equation*}
where agent $i$'s absolute risk aversion $-u_i''(x)/u_i'(x)$ is decreasing in $x$. DARA is a weak assumption in that it is satisfied by many commonly-used utility functions (like CRRA) and violation of DARA imply, counterfactually, that the rich spend less on risky assets. DARA yields uniqueness in this setting because it ensures that the derivative of good $2$ consumption with respect to wealth, $\partial x_{i2} / \partial w_i$, covaries positively with demand $x_{i2}$. With DARA, high consumption corresponds to a low sensitivity of marginal utility and thus a high sensitivity of demand to changes in wealth. Consequently, the natural buyers, whose income effects go in the ``right'' direction, dominate the aggregate income effect from a price increase. The result is that market demand is downward-sloping at any equilibrium, which implies uniqueness by continuity. However, if absolute risk aversion is increasing, as in the quadratic utility example from Section \ref{sec:anything}, the natural sellers dominate the aggregate income effect, opening the door for aggregate Giffen behavior and multiple equilibria.

The proof in \citet{GeanakoplosWalsh2018Uniqueness} relies on agents having the same endowment, which ensures demand has the same order as excess demand across agents, but they show that if Bernoulli utility is CRRA, uniqueness still holds with endowment heterogeneity, as long as the endowments are ordered according to utility weights. Specifically, patient agents must initially have more of the first good and less of the second good.

But it turns out there is more to say about uniqueness in the CRRA setting with two goods and risk aversion greater than one (elasticity of substitution $\varepsilon<1$), and \citet{Won2023} provides the definitive analysis. First, he shows that individual CRRA demand has two inflection points as the price of good $2$ rises.\footnote{As above, the price of good $1$ is normalized to one.} The first inflection price is where the income effect begins to dominate and demand becomes upward-sloping, and the second price is where the income effect maximally dominates, after which demand asymptotes to the endowment $e_{i2}$ (so demand is strictly quasi-convex in the price).\footnote{\citet{Gimenez2022} derives a similar result for general, non-separable, strictly quasi-concave utility with two goods. Specifically, he shows that relative risk aversion greater than one implies the goods are gross complements in some price range.} Second, he shows how to bound the set of equilibrium prices. The bounds are based on the extrema of two types of artificial economies: the economies consisting of the individual agents in autarky and the economies where a representative agent has the utility function of one of the agents. Third, realizing that the location of the individual inflection points relative to the price bounds determines the shape of market demand over the possible equilibrium price range, \citet{Won2023} derives several sufficient conditions for uniqueness. In the spirit of (but more general than) the \citet{GeanakoplosWalsh2018Uniqueness} CRRA result, each condition takes the form of an agent-specific upper bound on $e_{i2}/e_{i1}$, that is, individual endowment growth. The upper bounds depend only on the individual's utility parameters and the maximum of the price range. Corollary 5 is the most stark result: if either
\begin{enumerate*}
    \item the agents have the same substitution elasticity $\varepsilon$ or
    \item the most patient agent is sufficiently impatient (see the paper for details),
\end{enumerate*}
then the following simple condition (for all $i \in I$) ensures uniqueness:
\begin{equation}\label{eq:Won}
 \frac{e_{i2}}{e_{i1}} < \frac{\varepsilon_i}{1-\varepsilon_i} \left( \frac{1}{\rho_{\max}} \left( \frac{e_2}{e_1} \right)^{1/\varepsilon_i} + \frac{1}{\varepsilon_i} \left( \frac{\rho_i}{\rho_{\max}} \right)^{\varepsilon_i} \frac{e_2}{e_1}  \right),   
\end{equation}
where $\rho_i = \beta_{i2}/\beta_{i1}$ is the patience (relative utility weight) of agent $i$ and $\rho_{\max} = \max_i\rho_i$ is the patience of the most patient agent. In short, if the condition holds for all $i$, individual Giffen behavior lies outside of the equilibrium price bounds, which depend on the (easily calculated) extrema of the artificial and representative agent economies. Interestingly, as risk aversion approaches one for all agents ($\varepsilon_i\uparrow 1$), the right-hand side of \eqref{eq:Won} tends to infinity and hence the inequality is satisfied, providing another illustration of the \citet{HensLoeffler1995} $\RRA \le 1$ result.

Finally, \citet{Won2023} extends his findings to an arbitrary number of goods, specifically the case with identical risk aversion and identical utility weights, except for the first good. The natural interpretation is that good $1$ is consumption today, the second period endowment is random with common probabilities, and patience is heterogeneous (so the goods $2,\dots,J$ are state-dependent consumption in different states of the world in the second period). In this case, the second period ``bundles'' together in a unique fashion (by standard constant elasticity of substitution (CES) analysis), and the economy reduces to \citet{Won2023}'s original two good setting.

\citet{Gimenez2022}, in the \emph{Journal of Mathematical Economics}, provides a comprehensive analysis for the case with two agents and two goods, adopting a novel ``offer curve'' approach and deriving sufficient conditions for uniqueness based on the shapes of offer curves.\footnote{With two goods, the offer curve for an agent is the set of demand points mapped out in the Edgeworth box as one varies the relative price.} While analysis is restricted to two agents, \citet{Gimenez2022} allows for general convex preferences that, for example, need not be separable. One of the main results is that we get uniqueness if, for at least one of the goods, both agents' offers curves are monotonic in the direction of that good. So his theorem allows offer curves to ``bend back'' in the direction of one good (corresponding to Giffen behavior), but then both curves must be monotonic in the direction of the other good. When the condition of the theorem fails but the agents' offers curves are monotonic in opposite directions, \citet{Gimenez2022} derives many uniqueness conditions, (in a similar spirit to \citet{Won2023}) that jointly restrict preference and endowments based on the location of inflection points. He also provides many results with separable utility, for example generalizing \citet{GeanakoplosWalsh2018Uniqueness} (in the two agent case).

\citet{Chipman2010} likewise considers the case with two goods and two agents. With identical CRRA Bernoulli utility and symmetric endowments and utility weights ($e_{11} = e_{22}$, $e_{12}=e_{21}$, $\beta_{11}=\beta_{22}$, $\beta_{12}=\beta_{21}$), equilibrium is unique if $\RRA \le 2$.

Finally, Andrea Loi and Stefano Matta have a sequence of recent articles related to uniqueness, three of which appear in the \emph{Journal of Mathematical Economics}. \citet{LoiMatta2018} and \citet{LoiMattaUccheddu2023} are based on \citet{Balasko1980} (all in the \emph{Journal of Mathematical Economics}). \citet{Balasko1980} shows that
\begin{enumerate*}
    \item if the commodity space is $\R^J$ (vs. $\R^J_{++}$ in our motivating example) and
    \item if there is a unique equilibrium for any distribution of endowments (fixing the total),
\end{enumerate*}
then equilibrium prices do not depend on the distribution of endowments. \citet{LoiMatta2018} and \citet{LoiMattaUccheddu2023} prove the converse for the two good and two agent cases, respectively. Their sufficient condition for uniqueness is, specifically, that the equilibrium manifold\footnote{The set of pairs of prices and endowments for which markets clear.} has zero curvature.\footnote{See also \citet{LoiMatta2021} in the \emph{Journal of Mathematical Economics}.}

Finally, \citet{LoiMatta2023} answer a question posed in \citet{TodaWalsh2017ETB}: are multiple equilibria possible with CRRA utility and two agents/goods if $1<\RRA \le 2$? The authors show the answer is ``no,'' while proving a much more powerful result under the assumption of hyperbolic absolute risk aversion (HARA) Bernoulli utility, in particular,
\begin{equation}\label{eq:hara}
u(x) = \frac{\gamma}{1-\gamma} \left( b + \frac{a}{\gamma} x  \right)^{1-\gamma}    
\end{equation}
with $0<\gamma \neq 1$, $a>0$, and $b \ge 0$, which satisfies DARA. \citet{LoiMatta2023} prove, with a new method using Newton's symmetric polynomials and Descartes' rule of signs, that $\gamma \in \left(1,\frac{I}{I-1} \right]$ is sufficient for uniqueness with two goods, arbitrary utility weights, and common HARA Bernoulli utility \eqref{eq:hara}. Since $I$ is the number of agents and CRRA is included in \eqref{eq:hara}, this answers the question of \citet{TodaWalsh2017ETB}.

\subsection{Finance and Incomplete Markets}\label{sec:finance}

Closely related to the papers we have discussed thus far is general equilibrium with incomplete markets (GEI). See \citet{Geanakoplos1990} (\emph{Journal of Mathematical Economics}) for an early survey. In our example \eqref{eq:opt}, GEI replaces the budget constraint $p \cdot x_i = p \cdot e_i$ with
\begin{subequations}
\begin{align}
    q \cdot \theta_i & = 0 \label{eq:GEI1} \\
    x_i &= e_i + A\theta_i \label{eq:GEI2},
\end{align}
\end{subequations}
where $\theta_i \in \R^L$ is $i$'s portfolio of the $L$ assets (which are zero net supply) and $q$ are the asset prices. Recalling that $J$ is the number of goods, $A$ is the $J \times L$ matrix of exogenous asset payoffs. Equilibrium is defined to be asset prices and portfolios such that portfolios maximize utility over consumption given prices and the asset market clears: $\sum_i \theta_i = 0$. Typically, the utility weights $\beta_{ij}$ are interpreted as probabilities and the goods as consumption in different states of nature, meaning each column of $A$ is the state-dependent payoffs of an asset (a column of ones is the risk-free asset). Note that it is trivial to
\begin{enumerate*}
    \item allow for a positive supply of assets or
    \item create a two-period model by adding a good $j=0$ and setting the first row and column of $A$ to $(1,0,0,\dots)$.
\end{enumerate*}
Usually, it is assumed the columns of $A$ are linear independent, meaning there are no redundant assets. So if $L=J$, the market is complete and the economy is equivalent to the one we originally described in example \eqref{eq:opt}. However, when $L<J$, the market is incomplete, and there are (potentially Pareto-improving) trades that cannot take place between the agents because the assets do not span the space of goods.

As extensively surveyed in \citet{HensPilgrim2002} (see their very useful Figure 6.1), versions of many of the classic complete markets uniqueness results carry over to GEI, and we discuss a few here related to the \emph{Journal of Mathematical Economics}.\footnote{There are also uniqueness papers in the more-specialized CAPM literature. See, for example, \citet{Dana1999} and \citet{HensLaitenbergerLoeffler2002}, both in the \emph{Journal of Mathematical Economics}.} Most of the proofs are quite similar to the complete markets ones, but they usually require some additional assumptions, especially on the structure of the payoff matrix $A$. For the case with separable utility and an infinite number of goods (\ie, states), \citet{Dana1995} shows the \citet{MitjushinPolterovich1978} result (uniqueness with collinear endowments and $\RRA <4$) holds if the endowments are tradeable, that is, $e_i$ is in the span of $A$. For a finite number of goods, \citet{HensPilgrim2002} show one does not even need separability for this result. Also with separability, \citet{Bettzuege1998} is able to establish \citet{MitjushinPolterovich1978} in GEI, but dropping endowment tradeability and replacing it with his alternate ``fundamental set'' assumption.

\citet{Dana1995} further shows that \citet{Eisenberg1961}--\citet{Chipman1974} (uniqueness with homothetic utility and collinear endowments) holds in her GEI setting with endowment tradeability. Also, with a finite number of goods and tradeable endowments, \citet{Eisenberg1961}--\citet{Chipman1974} holds without separability: see \citet[Theorem 1]{TodaWalsh2020}. Using this uniqueness result, \citet[Theorem 2]{TodaWalsh2020} derive comparative statics results between the wealth distribution and the equity premium.

\section{Uniqueness in Macroeconomics}\label{sec:macro}

Much of macroeconomics employs a representative agent or considers economies in which equilibrium is equivalent to the solution to a planning problem. In both cases, uniqueness is often the natural outcome. With a representative agent, this is because market excess demand will satisfy WARP, and we know in advance that the agent will have to consume all resources in equilibrium, which implies a unique set of prices through the agent's necessary conditions for optimality. And concave planning problems with convex constraints have unique solutions, so we can also get uniqueness if equilibrium reduces to a particular planning problem.

Heterogeneous agent macroeconomic models are, however, a different story, unless you can combine the budget constraints and write the economy in date-0 Arrow-Debreu form, in which classic results from GE theory may be available \citep{KehoeLevine1985}. With incomplete markets, random income, and credit frictions, one can quickly reenter the world of anything goes.

\subsection{Overlapping Generations Models}

In overlapping generations (OLG) models, not only multiplicity of equilibria but also equilibrium indeterminacy is possible, meaning the existence of a continuum of equilibria in a neighborhood of a particular equilibrium. This possibility was first pointed out by \citet{Gale1973}, who studies the classic \citet{Samuelson1958} endowment economy. To illustrate this result, consider the classical two-period overlapping generations model with a single perishable good. Let $U(y_t,z_{t+1})$ be the utility function of generation $t$, where $(y_t,z_{t+1})$ is consumption when young and old. For simplicity, suppose $U$ is additively separable, so
\begin{equation*}
    U(y,z)=u(y)+\beta v(z),
\end{equation*}
where $u,v$ are the utility of the young and old and $\beta>0$ is the discount factor. We assume $u$ is twice continuously differentiable, $u'>0$, $u''<0$, satisfies the Inada condition $u'(0)=\infty$, and likewise for $v$. Each period, the young and old have endowments $(a,b)$, where $a>0$ and $b\ge 0$. In addition, the initial old are endowed with an intrinsically useless asset (\ie, an asset paying no dividends like fiat money) in unit supply.

Letting $P_t\ge 0$ be the price of the asset at time $t$ in units of the consumption good, the budget constraints of generation $t$ are
\begin{align*}
    &\text{Young}: & y_t+P_tx_t&=a,\\
    &\text{Old}: & z_{t+1}&=b+P_{t+1}x_t,
\end{align*}
where $x_t$ is asset holdings. As usual, a competitive equilibrium is defined by a sequence $\set{(P_t,x_t,y_t,z_t)}_{t=0}^\infty$ such that
\begin{enumerate*}
    \item each generation maximizes utility subject to the budget constraints,
    \item commodity markets clear, so $y_t+z_t=a+b$,
    \item asset markets clear, so $x_t=1$.
\end{enumerate*}

In any equilibrium, the budget constraints imply the equilibrium allocation $(y_t,z_{t+1})=(a-P_t,b+P_{t+1})$. Therefore the equilibrium is completely determined by the price sequence $\set{P_t}_{t=0}^\infty$. It is easy to show that either $P_t=0$ for all $t$ or $P_t>0$ for all $t$ \citep[\S2]{HiranoToda2024JME}. We introduce some terminology to classify equilibria. We say that an equilibrium is
\begin{enumerate*}
    \item \emph{stationary} or a \emph{steady state} if $P_t=P$ is constant,
    \item  \emph{non-monetary} (or \emph{fundamental}) if $P_t=0$ for all $t$,
    \item \emph{monetary} (or \emph{bubbly}) if $P_t>0$ for all $t$,
    \item \emph{asymptotically non-monetary} if $\liminf_{t\to\infty}P_t=0$,
    \item \emph{asymptotically monetary} if $\liminf_{t\to\infty}P_t>0$.
\end{enumerate*}
By definition, non-monetary equilibria are asymptotically non-monetary, and asymptotically monetary equilibria are monetary. However, there could be monetary equilibria that are asymptotically non-monetary.

Because the non-monetary equilibrium is obviously unique ($P_t\equiv 0$), we focus on monetary equilibria. Take any monetary equilibrium, so $P_t>0$ for all $t$. Using the budget constraints to eliminate $(y_t,z_{t+1})$, generation $t$ seeks to maximize
\begin{equation*}
    u(a-P_tx_t)+\beta v(b+P_{t+1}x_t).
\end{equation*}
Taking the first-order condition and imposing the market clearing condition $x_t=1$, we obtain the equilibrium condition
\begin{equation}
    -u'(a-P_t)P_t+\beta v'(b+P_{t+1})P_{t+1}=0. \label{eq:eqcond}
\end{equation}
The following lemma provides a necessary and sufficient condition for the existence of a monetary steady state.

\begin{lem}\label{lem:ss_bubbly}
There exists a monetary steady state if and only if $u'(a)<\beta v'(b)$.
\end{lem}

\begin{proof}
If a monetary steady state $P>0$ exists, setting $P_t=P_{t+1}=P$ in \eqref{eq:eqcond} and dividing by $P>0$, we obtain
\begin{equation*}
    \psi(P)\coloneqq -u'(a-P)+\beta v'(b+P)=0.
\end{equation*}
Since $\psi'(P)=u''(a-P)+\beta v''(b+P)<0$, it follows that
\begin{equation*}
    0=\psi(P)<\psi(0)=-u'(a)+\beta v'(b)\implies u'(a)<\beta v'(b).
\end{equation*}

Conversely, suppose $u'(a)<\beta v'(b)$. Then
\begin{equation*}
    \psi(0)=-u'(a)+\beta v'(b)>0>-\infty=-u'(0)+\beta v'(b+a)=\psi(a).
\end{equation*}
By the intermediate value theorem, there exists $P>0$ such that $\psi(P)=0$. Furthermore, since $\psi$ is strictly decreasing, such a $P$ is unique. Therefore there exists a unique monetary steady state.
\end{proof}

We next study the local determinacy of equilibria around the steady states. To this end, we write the equilibrium condition \eqref{eq:eqcond} as $\Phi(P_t,P_{t+1})=0$, where
\begin{equation}
    \Phi(\xi,\eta)=-u'(a-\xi)\xi+\beta v'(b+\eta)\eta. \label{eq:Phi}
\end{equation}
To apply the implicit function theorem, we compute the partial derivatives
\begin{subequations}\label{eq:Phi_partial}
\begin{align}
    \Phi_\xi(\xi,\eta)&=-u'(a-\xi)+u''(a-\xi)\xi, \label{eq:Phi_xi}\\
    \Phi_\eta(\xi,\eta)&=\beta v'(b+\eta)+\beta v''(b+\eta)\eta. \label{eq:Phi_eta}
\end{align}
\end{subequations}
Therefore if $\Phi_\eta\neq 0$, we may locally solve $\Phi(\xi,\eta)=0$ as $\eta=\phi(\xi)$, where
\begin{equation}
    \phi'(\xi)=-\frac{\Phi_\xi}{\Phi_\eta}=\frac{u'(a-\xi)-u''(a-\xi)\xi}{\beta v'(b+\eta)+\beta v''(b+\eta)\eta}. \label{eq:phi'}
\end{equation}
Let $P$ be a steady state. If $\abs{\phi'(P)}\neq 1$, the Hartman-Grobman theorem \citep[Theorem 4.6]{Chicone2006} implies that the local behavior of the dynamical system $\Phi(P_t,P_{t+1})=0$ (or $P_{t+1}=\phi(P_t)$) is qualitatively the same as that of the linearized system
\begin{equation*}
    P_{t+1}-P=\phi'(P)(P_t-P).
\end{equation*}
Therefore if $\abs{\phi'(P)}>1$, then the unique $\set{P_t}$ converging to $P$ is $P_t=P$ and the equilibrium is \emph{locally determinate}; if $\abs{\phi'(P)}<1$, then for any $P_0$ sufficiently close to $P$, there exists $\set{P_t}$ converging to $P$ and the equilibrium is \emph{locally indeterminate}.

The following proposition characterizes the local determinacy of the non-monetary steady state.

\begin{prop}[\citealp{Gale1973}, Theorem 4]\label{prop:b>0}
Let $b>0$.
\begin{enumerate}
    \item If $u'(a)>\beta v'(b)$, the non-monetary steady state is locally determinate, \ie, if $\set{P_t}_{t=0}^\infty$ is an equilibrium with $P_t\to 0$, then $P_t=0$ for all $t$.
    \item If $u'(a)<\beta v'(b)$, the non-monetary steady state is locally indeterminate, \ie, for any $P_0>0$ sufficiently close to 0, there exists a monetary equilibrium $\set{P_t}_{t=0}^\infty$ with $P_t\to 0$.
\end{enumerate}
\end{prop}

\begin{proof}
Setting $\xi=\eta=P=0$ in \eqref{eq:phi'}, we obtain $\phi'(0)=u'(a)/\beta v'(b)>0$. The claim follows from the definition of local (in)determinacy.
\end{proof}

Combining Lemma \ref{lem:ss_bubbly} and Proposition \ref{prop:b>0}, we conclude that whenever a monetary steady state exists, there exists an asymptotically monetary equilibrium as well as a continuum of monetary but asymptotically non-monetary equilibria.

We next study the local determinacy of the monetary steady state. The following proposition shows that the monetary steady state is locally determinate whenever the relative risk aversion of the old is sufficiently low.

\begin{prop}\label{prop:locdet_b}
Let $u'(a)<\beta v'(b)$. Then the monetary steady state $P>0$ is locally determinate if $\gamma_v(b+P)<1+b/P$, where $\gamma_v(z)\coloneqq -zv''(z)/v'(z)>0$ denotes the relative risk aversion of $v$.
\end{prop}

\begin{proof}
Setting $\xi=\eta=P>0$ in \eqref{eq:Phi_partial} and using the equilibrium condition \eqref{eq:eqcond}, at the monetary steady state we obtain
\begin{equation*}
    \Phi_\xi+\Phi_\eta=u''(a-P)P+\beta v''(b+P)P<0.
\end{equation*}
If $\Phi_\eta>0$, dividing both sides by $\Phi_\eta$, we obtain
\begin{equation*}
    0>\frac{\Phi_\xi}{\Phi_\eta}+1=-\phi'(P)+1\iff \phi'(P)>1,
\end{equation*}
implying local determinacy. Therefore it suffices to show $\Phi_\eta>0$. Dividing \eqref{eq:Phi_eta} by $\beta v'(b+\eta)>0$ and setting $\eta=P>0$, we obtain
\begin{equation*}
    \frac{\Phi_\eta}{\beta v'(b+P)}=1-\gamma_v(b+P)\frac{P}{b+P}.
\end{equation*}
Therefore $\Phi_\eta>0$ if $\gamma_v(b+P)<1+b/P$.
\end{proof}

Obviously, determinacy of the monetary steady state obtains under the condition $\RRA\le 1$. Equilibrium indeterminacy in OLG models is a common feature \citep{GeanakoplosPolemarchakis1991}. However, there are a few uniqueness results. \citet{BalaskoShell1981_3} show that, with Cobb-Douglas utility, equilibrium is unique without money and the equilibrium manifold is one-dimensional (in particular, a continuum) with money (if nonempty). \citet{KehoeLevineMas-CollelWoodford1991} prove many results under the assumption of the gross substitutes property, including
\begin{enumerate*}
    \item uniqueness of equilibrium if the aggregate endowment has finite value (Theorem A),
    \item uniqueness of monetary and non-monetary steady states (Theorem C),
    \item convergence to steady states (Theorem D), and
    \item uniqueness of equilibrium converging to monetary steady state (Theorem F).
\end{enumerate*}

\subsection{Bewley Models}

We next consider the so-called canonical ``Bewley'' model \citep{Bewley1986,Imrohoroglu1989,Huggett1993,Aiyagari1994}, also known as the ``Bewley-\.{I}mrohoro\u{g}lu-Huggett-Aiyagari'' model (or some permutation of a subset of these names). Variants of these models have been widely used for policy analysis, understanding inequality, and asset pricing.\footnote{See \citet{CherrierDuarteSaidi2023} for a recent review.} In the Bewley model, each of a large number of (ex ante identical) agents solves a recursive problem like
\begin{subequations}
\begin{align}\label{eq:bewley}
    & V(k,x) =  \max_{c\ge 0,k' \ge -b} \set{ u(c) + \beta \E_x \left[ V(k',x') \right] } \\
    & \st~c + k' =wy(x) + Rk,
\end{align}
\end{subequations}
where $V$ is the value function, $u$ is the Bernoulli utility function, $\beta$ is the discount factor, $k$ is savings, $b$ is a borrowing limit, $R$ is the gross interest rate, and $w$ is the wage. $y(x)$ is the agent's exogenous labor supply, and the only randomness stems from the evolution of the agent-specific state variables $x$ ($\E_x$ is expectation conditional on $x$). Typically, either the asset $k$ is in zero net supply (with exogenous labor income) and $R$ clears the market in equilibrium, as in \citet{Huggett1993}, or a representative firm rents labor and capital from the agents and $R$ and $w$ clear the capital/labor markets in equilibrium, as in \citet{Aiyagari1994}. While in general the prices depend on the cross-sectional distribution of assets and agent-specific state variables, the literature usually focuses on stationary equilibria (with constant prices), in which the wealth and income distributions are constant (even though there is churn within them). Unlike in our example \eqref{eq:opt}, the market here is incomplete because the borrowing constraint and small number of assets (one) prevent the agents from trading all (state- and time-dependent) goods.

Establishing uniqueness or constructing examples with multiple equilibria is challenging in this context: closed-form solutions are not usually available for Bewley models, and most papers solve them with numerical methods. Fortunately, there are some new results in the recent literature, but these include examples with even multiple \emph{stationary} equilibria. \citet{Toda2017JEDC} shows how to construct \citet{Huggett1993}-type economies with multiple stationary equilibria in the case of CARA Bernoulli utility without a borrowing constraint. This is perhaps surprising on account of the findings of \citet{HensLoeffler1995} reviewed in Section \ref{sec:classic}. More specifically, \citet{Toda2017JEDC} shows that the particular income process is important for the number of equilibria. If $y(x)$ follows an AR(2) or ARMA(1,1) process, then you can construct examples with at least three stationary equilibria. On the other hand, \citet{Toda2017JEDC} provides two cases where the stationary equilibrium is unique in the CARA \citet{Huggett1993} economy without a borrowing constraint:
\begin{enumerate*}
    \item AR(1) labor income with an arbitrary shock distribution (which is the model of \citet{Wang2003}) and
    \item Gaussian persistent-transitory (PT) labor income with non-negative correlation between the components.
\end{enumerate*}
The PT case, in which labor income is the sum of an AR(1) and white noise, is widely-used in applied quantitative work.\footnote{\citet{Toda2017JEDC} shows, however, that multiple equilibria are possible if the PT components are negatively correlated. Note also that \citet{Toda2017JEDC}'s uniqueness results do not rule out multiple nonstationary equilibria, which \citet{Calvet2001} shows can arise in even more restricted settings such as i.i.d.\ Gaussian income.} Some earlier papers, in a different context, also explore uniqueness in the CARA case. \citet{Calvet2001} examines a CARA-Gaussian GEI economy without aggregate risk and with many time periods and shows there is a unique equilibrium with either a finite horizon or complete markets. However, his model features equilibrium indeterminacy with infinite horizon and incomplete markets due to the lack of terminal condition at $T=\infty$. \citet{AngeletosCalvet2005} (\emph{Journal of Mathematical Economics}) extend \citet{Calvet2001}'s model with production and show by computing the eigenvalues of the linearized system around the steady state, the equilibrium can be locally determinate or indeterminate (\ie, there exists a continuum of equilibrium paths converging to the steady state) depending on the model parameters.

In the \citet{Aiyagari1994} version of the Bewley model with Cobb-Douglas production and CRRA Bernoulli utility, \citet[Figure 1]{Acikgoz2018} presents a parameterization that yields a numerical example with multiple stationary equilibria, although he notes the parameter values are ``rather extreme'' (\eg, $\beta = 0.1$, full depreciation of capital). Nevertheless, he explains that multiplicity stems from the familiar income effect domination outlined in Section \ref{sec:anything} (and not a general equilibrium effect through the wage). With incomplete markets, the agents heavily weight worst-case scenarios, and permanent income in the worst-case scenario is increasing in the interest rate for most agents in the ergodic distribution. Therefore, depending on the parameters, aggregate capital supply (from the agents to the firm) may be declining in the interest rate as many agents increase consumption today in response to a rising interest rate. Since firm capital demand is always decreasing in the interest rate, the downward-sloping supply can induce multiplicity.\footnote{\citet{Kirkby2019}, on the other hand, argues there is near certainty the stationary equilibrium in \citet{Aiyagari1994} is unique for standard calibrations: his algorithm theoretically converges to all equilibria in this context, and he finds only one in each case with standard parameters.} 

While we have seen that uniqueness from CARA does not generally carry over to the Bewley model setting, at least in the \citet{Huggett1993} model with loose credit, one assumption from the complete markets general equilibrium literature does: $\RRA \le 1$. \citet{Light2020} shows, for the \citet{Aiyagari1994} model, that constant returns to scale (CRS) production with elasticity of substitution (ES) $\ge 1$ (\eg, Cobb-Douglas) and CRRA utility with $\RRA \le 1$ are sufficient for a unique stationary equilibrium. In short, $\RRA \le 1$ is sufficient for substitution effect dominance (as with complete markets) and $\ES \ge 1$ prevents any complementarity stemming from production (as implied in \citet{Acikgoz2018} for Cobb-Douglas). \citet{Light2023} also proves a version of his result in a \citet{Huggett1993} model with many goods, CES utility, i.i.d.\ endowments, and the natural borrowing limit, showing there is a unique stationary equilibrium if $\ES \ge 1$.

$\RRA \le 1$ is also sufficient for uniqueness in \citet{Toda2019JME}. He studies a \citet{Huggett1993}-type model with random (Markov) discount factors, constant income,  no credit limit (besides the natural one), and the perpetual youth model of life/death (to get a stationary wealth distribution, which he proves has a Pareto tail). Under sufficient impatience and CRRA utility with $\RRA \le 1$, there exists a unique stationary equilibrium.

Finally, \citet{AchdouHanLasryLionsMoll2022} show that even without the CRRA form, $\RRA \le 1$ is sufficient for a unique stationary equilibrium in their continuous time \citet{Huggett1993} model without borrowing. This result carries through to the continuous \citet{Aiyagari1994} model with CRS production, but as in \citet{Light2020}, one more assumption is needed: the ES must be strictly larger than the capital share of output.

The condition $\RRA \le 1$ also appears in infinite horizon models with limited risk-sharing due to collateral constraints or ``limited pledgeability'' of endowments. \citet{BloiseCitanna2015}, in the \emph{Journal of Mathematical Economics}, build on \citet{GottardiKubler2015} and prove $\RRA \le 1$ implies uniqueness in an infinite horizon endowment economy in which the only friction is limited commitment (markets are otherwise complete).

While it appears $\RRA \le 1$ (and the resulting dampening of income effects) is a fairly robust sufficient condition for uniqueness across heterogeneous agent macro models with various frictions, and CARA works in some settings, the field needs more results. $\RRA>1$ is widely-used in applied macro models, but at the same time, there are very few non-uniqueness examples. Given the computational nature of the field, one wonders whether most papers are calculating a unique equilibrium, or if existing numerical methods are simply missing the multiplicity.\footnote{The working paper by \citet{Proehl2018} suggests there is hope in establishing uniqueness of recursive equilibrium under quite general conditions in Bewley models with \emph{aggregate} risk, such as \citet{KrusellSmith1998}.}

\section{Trade and Production}

Thus far we have focused on endowment economies, except with respect to some of the macro papers in Section \ref{sec:macro}, for which the issue of uniqueness was confounded with market incompleteness and credit frictions. But simply introducing production into our frictionless, complete markets example \eqref{eq:opt} unfortunately still renders uniqueness much more elusive relative to the endowment economy case. \citet{Kehoe1984}, for example, provides a three-equilibrium example with constant returns production, four agents/goods, and \emph{Cobb-Douglas} utility.\footnote{See also \citet{Kehoe1985}.}


The basic issue is that with production, the supply of goods also moves in response to price changes, and firms' factor demands introduce a new income effect channel (through, for example, wages or capital income). This opens the door for highly non-monotonic market excess demands, even in cases where the excess demand is well-behaved for fixed endowments. We have already discussed several examples of non-uniqueness with production \citep{Kehoe1984,Kehoe1985,AngeletosCalvet2005,Acikgoz2018}, and the uniqueness results of \citet{Light2020} and \citet{AchdouHanLasryLionsMoll2022} require restrictions on production elasticities, even with agent $\RRA \le 1$.

The general equilibrium literature on uniqueness with production is perhaps thinner compared with the endowment economy case. Nonetheless, there are many earlier papers on this topic, for example the works by Timothy Kehoe in the 1980s, and the surveys of \citet{Kehoe1998}, \citet{Mas-Colell1991}, and \citet{MWG1995} provide citations. While the generality considered in this earlier literature reaches somewhat pessimistic conclusions,\footnote{\citet{KehoeWhalley1985}, for example, write:
\begin{quote}
    ``It remains an open problem whether or not non-uniqueness is largely a theoretical curiosum, similar in practical significance to the possibility of the simplex algorithm requiring an exponential number of steps to solve a linear programming problem, or whether it is a prevalent feature even in empirically based general equilibrium models. The difficulty is that it is next to impossible to establish uniqueness of equilibrium for large dimensional models since no algorithm is known that finds all of the equilibria of a model.''
    \end{quote}} 
the more applied fields of trade, production networks, and spatial models have recently made much progress. Indeed, one of the main applications of general equilibrium theory is international trade, which we have not yet discussed in this review. Large-scale equilibrium models with production and dozens of regions and industries are now widely-used for policy analysis.\footnote{See, for example, \citet{whatif2022}.} As the issue of uniqueness is tightly linked to the usefulness of comparative static analysis, trade theory has recently become a key field for the study of uniqueness of competitive equilibrium with production.

One of the earliest uniqueness results with a production network is the ``non-substitution theorem,'' often attributed to \citet{Samuelson1951}. Suppose each of the $J$ goods is produced by a distinct competitive firm (or location). Each firm $j$ produces it good according to the CRS production function $Y^j = F^j(y^j_1,\dots,y^j_J,\ell^j)$, where $y^j_k$ is the input of good $k$ and $\ell^j$ is the input of labor. The key assumptions (other than CRS) are that
\begin{enumerate*}
    \item each good is specific to a particular producer and
    \item there is only one ``factor,'' which we have called labor but could be interpreted as capital or land.
\end{enumerate*}
Agent income, the right-hand-side of the budget constraint in example \eqref{eq:opt}, becomes $wL^i$, where $w$ is the wage (to be determined in equilibrium) and $L^i$ is an exogenous supply of labor. Equilibrium with production is defined to be consumption allocations $x^i$, goods prices $p$, wage $w$, firm inputs $(y^j_k,\ell^j)$, and production $Y^j$ such that:
\begin{enumerate*}
    \item given prices, consumption choices are optimal,
    \item firms maximize profits given prices (implying profits are zero by CRS),
    \item the labor market clears, $\sum_i L^i = \sum_j \ell^j$, and
    \item the goods markets clear, $\sum_i x^i_j + \sum_k y^k_j = Y^j$, where output $Y_j$ is given by the production functions.
\end{enumerate*}

The basic version of the non-substitution theorem says that under these conditions, relative prices (and thus firm input ratios) are uniquely determined by the supply side of the economy. Since the demands of the agents are unique given prices (including the wage), the competitive equilibrium with production is also unique. A concise proof of the uniqueness of relative prices is provided in \citet{Stiglitz1970}, which we specialize here to CES for clarity and briefly sketch. In the CES case for $F^j$, marginal cost pricing takes the form
\begin{equation*}
    p_j = c^j(w,p) \coloneqq \left(a^j_\ell w^{1-\sigma^j} + \sum_{k \in J} a^j_k p_k^{1-\sigma^j} \right)^\frac{1}{1-\sigma^j},
\end{equation*}
where the exogenous production weights, the $a$'s ($>0$), sum to one, and $\sigma^j>0$ is the elasticity of substitution. Note that $c^j$ is concave and positive. Using the normalization $w =1$ (instead of $p_1 = 1$ from earlier) and stacking the $c^j$'s to form $c$, equilibrium prices satisfy $p = c(1,p)$. Following \citet{Stiglitz1970}, the non-substitution theorem is equivalent to $p = c(1,p)$ having a unique solution. Suppose $p^*$ satisfies $p^* = c(1,p^*)$. Taking a first order approximation of $c^j(1,p)$ at $p^*$, we can form a plane $p_j = b^j_\ell + \sum_k b^j_k p_k$ tangent to $c^j$ at $p^*$, where the (strictly positive) $b$'s correspond to output shares in production. Considering all $j \in J$, in matrix notation we have $p = b_\ell + Bp$. Since both $p - Bp$ and $B$ are positive, $B$ is ``Leontief'' (or ``productive'') and, hence, $I-B$ is invertible. We have $p^* = (I-B)^{-1}b_\ell$. Now suppose $p^{**}$ is also an equilibrium. By the concavity of $c^j$, $p^{**} = c(1,p^{**}) \le  b_\ell + Bp^{**}$ (points other than $p^*$ are ``below'' the plane), so it follows $p^{**} \le p^*$. But we could have formed the plane at $p^{**}$, so we also have $p^{**} \geq p^*$, implying the prices are unique.

This specific non-substitution argument does not, however, extend to the case with many factors (\eg, the agents or ``countries'' have differentiated labor supplies). The reason is that the supply side alone is unable to determine the relative factor prices. While goods prices are unique here given wages, the determination of wages depends also on the demand side, and the usual assumptions for uniqueness on the demand side will not necessarily suffice due to additional feedback: changing wages lead to changing demands and goods prices, which in turn affect labor demand and hence wages.

There are several other version of the non-substitution theorem, but we now move forward in time to highlight more recent results. \citet{AcemogluAzar2020} provide an important generalization of \citet{Samuelson1951}'s non-substitution theorem. Specifically, they allow general frictions (which could represent taxes, credit frictions, markups, etc.) and, moreover, allow the CRS firms to choose their suppliers (the set of which affects productivity). So in \citet{AcemogluAzar2020}, the production network is \emph{endogenous}. Nonetheless, a proof with a flavor of \citet{Stiglitz1970}'s goes through, prices are uniquely pinned-down by the supply side, and the equilibrium quantities are generically unique (with some mild assumptions on utility).

\citet{ScarfWilson2005} provide a result complementary to the non-substitution theorem. In the Ricardian model of \citet{ScarfWilson2005} there are many countries, each with a fixed supply of labor that can produce any of the $J$ goods according to a CRS activity matrix. For a particular good, only the lowest cost countries produce and export. While there are many factors (the labor in the different countries), the firms only use local labor as an input (and not the goods). \citet{ScarfWilson2005} show that if the demand side satisfies gross substitutes, then equilibrium is unique,\footnote{See also \citet{AdaoCostinotDonaldson2017}.} which is perhaps surprising given the result of \citet{Kehoe1984}.

More recently, \citet{KucheryavyyLynRodriguezClare2023a} establish uniqueness in an ``Armington'' model with many countries and industries but where, as in \citet{ScarfWilson2005}, only local labor is used in production. ``Armington'' (in contrast with ``Ricardian'') means that consumers exogeneously demand industry bundles formed from other countries' goods according to CES (with elasticity $\sigma_k>1$ in each industry $k$). They allow for a scale externality (EES) governed by parameter $\phi_k$ in each industry $k$, so that industries with more labor become more productive. With Cobb-Douglas utility over industry bundles in each country, they prove that if trade is frictionless, there are two countries, or we have a small open economy, then equilibrium is unique if the EES parameter is less than inverse trade elasticity: $\phi_k \le 1/(\sigma_k -1)$.\footnote{The case with equality requires a condition on trade costs.}

The trade/production results we have discussed so far feature either
\begin{enumerate*}
    \item a production network of single-good producing firms with a single supply of labor shared globally or
    \item countries that produce many goods but only using local labor (so there are many, country-specific factors) but not intermediates.
\end{enumerate*}
But state-of-the-art trade models contain both features: a global production network \emph{and} country-specific factors of production. The difficulty in establishing uniqueness, as we described above, is that standard assumptions on utility and production do not necessarily control income effects stemming from the factor prices (\eg, wages).

Two key seminal works in quantitative trade theory are the Ricardian models of \citet{eaton-kortum2002} and \citet{CaliendoParro2015} (which is the multi-industry version of the former paper). In this type of model, consumers and producers around the world demand a continuum of intermediates, aggregated via CES, in each industry (in \citet{eaton-kortum2002} there is only one industry). Any of the countries, each of which has its own labor, consumers, and firms, can produce the continuum of different varieties, but only the lowest cost suppliers  produce and export (accounting for icebergs and tariffs). Assuming firms within a country draw productivities from Fr\'{e}chet distributions, each country ends up being the minimum cost seller in a positive mass of varieties, and equilibrium is ``as if'' the countries produced differentiated goods demanded by the other countries according to exogenous weights (as in an Armington model). While this setting features many factors of production (the labor in each country) and a global production network, \citet{AlvarezLucas2007} are able to provide sufficient conditions for uniqueness. Assuming
\begin{enumerate*}
    \item production of both final consumption (``nontradeables'') and intermediate goods (``tradeables'') are Cobb-Douglas over labor and the CES bundle of traded inputs,
    \item countries are identical except with respect to labor endowments, average productivity, and trade barriers,
    \item  trade barriers are not too high (with respect to conditions that depend on the elasticity and labor shares), and
    \item tariffs are not source specific,
\end{enumerate*}
then equilibrium is unique if the labor share for nontradables (final goods) is greater than the labor share in the tradeable intermediates sector.

Most of the papers we have discussed in this article approach the uniqueness question by attempting to establish gross substitutes or WARP for market demand or by studying the slope of demand local to equilibrium. In a recent sequence of papers, Treb Allen, Costas Arkolakis, and their coauthors take a somewhat different but highly fruitful approach. In short, they examine the spectral radius of a matrix bounding the model's elasticities. \citet{AllenArkolakisTakahashi2020} write down a general single industry model with many agents, the optimality and market clearing conditions of which subsume many trade and economic geography models, including \citet{eaton-kortum2002} (and its Armington analogue). As they show, in a large class of models with CES features, you get similar CES demands (the ``gravity'' equations) and supply curves (from the production functions) that, once applying market clearing, boil down to the same system of equations for prices. Generalizing the result of \citet{AllenArkolakis2014},\footnote{See also \citet{KucheryavyyLynRodriguezClare2023b}.} \citet{AllenArkolakisTakahashi2020} provide several sufficient conditions for uniqueness, the most important being that the demand elasticity is weakly negative and the supply elasticity is weakly positive with respect to input prices. In the context of trade models along the lines of ones we have already discussed, the result is beautifully stark: a CES demand parameter $\geq 1$  and a production labor share $\le 1$ are together sufficient for uniqueness. For example, in an Armington model where each country produces (via Cobb-Douglas) a particular good using a local labor endowment and and CES bundle of all goods (which is also consumed by the representative agent in each country), equilibrium is unique if the (common) CES elasticity is $\geq 1$ and the (common) labor share is $\le 1$.\footnote{This result requires some assumptions on icebergs, but they are weak.}

\citet{AllenArkolakisLi2023} substantially generalize \citet{AllenArkolakisTakahashi2020}, considering the equilibrium conditions (in their notation)
\begin{equation*}
    x_{ih} = \sum_{j=1}^N f_{ijh} (x_{j1},\dots,x_{jH}),
\end{equation*}
where $x_{ih}$ is the ``outcome'' for location $i \in N$ for interaction $h \in H$, and differentiable $f_{ijh}: \R^H_{++} \to \R_{++}$ governs the interactions (according to, say, production functions/aggregators). This setting covers versions of dynamic or multi-industry Armington or \citet{eaton-kortum2002} models, but those cases are a small subset of the admissible spatial models. The key object in the paper is the bounding matrix $\mathbf{A}$ with $\mathbf{A}_{hh'} \coloneqq \sup_{i,j} \abs{\frac{\partial \ln f_{ijh}}{\partial \ln x_{jh'}}} $, that is, the maximum \textit{cross-location} interaction strength. Their theorem says that equilibrium is unique when the spectral radius of $\mathbf{A}$ is $\le 1$, although the radius of one case (which arises in CES models) requires some additional assumptions on $f_{ijh}$.\footnote{For example, to apply the theorem to a multi-sector \citet{CaliendoParro2015}-type model in which the spectral radius of $\mathbf{A}$ is one, the theorem requires that labor shares do not vary by country. Adopting an approach related to (but distinct from) \citet{AllenArkolakisLi2023}, the working paper by \citet{BifulcoGlueckKrebsKukharskyy2022} makes progress on establishing uniqueness in trade models with location-specific shares and elasticities.} The key insight is that it is much easier to analyze the ``sufficient'' $H \times H$ matrix $\mathbf{A}$ than either the $N^2 \times H$ functions $f_{ijh}$ or the full $NH \times NH$ Jacobian. For example, in the one-industry \citet{AllenArkolakisLi2023} model, $\mathbf{A}$ is $2 \times 2$, even though there are many locations and goods.\footnote{See also \citet{AllenArkolakis2022} for uniqueness results that use a similar technique in the context of transportation and traffic.}

\section{Summary and Looking Forward}

Table \ref{tab:papers} lists examples, in chronological order, of sufficient conditions on preferences, production technology, and endowments that ensure the uniqueness of competitive equilibria. As is clear from the table, $\RRA\le 1$ (or equivalently $\ES\ge 1$) is a relatively robust sufficient condition for equilibrium uniqueness. However, this is a quite strong parametric restriction. The result of \citet{Eisenberg1961} (see also \citet{Chipman1974}, \citet{HensPilgrim2002}, and \citet{TodaWalsh2020}) is remarkable in that equilibrium uniqueness obtains under arbitrary heterogeneous homothetic preferences, although one needs to assume collinear endowments (in the context of financial markets, agents are endowed with some fraction of the index fund). 

As we have seen, establishing uniqueness with $\RRA>1$ and non-collinear endowments is possible with a small number of agents or goods. We hope that the joint restrictions on endowments and preferences from this more stylized setting will help researchers find uniqueness conditions applicable in quantitative macro and trade models. Another potential way forward is the further development and implementation of algorithms known to find all equilibria for a given set of parameters (see, for example, \citet{KublerSchmedders2010} and \citet{Kirkby2019}). Randomizing over realistic parameter ranges, this approach can tell us whether non-uniqueness is, echoing \citet{Kehoe1984} in the \emph{Journal of Mathematical Economics}, ``pathological'' or ``commonplace'' in applied settings.

\begin{table}[!htb]
	\centering
	\caption{Uniqueness Results}
	\resizebox{\textwidth}{!}{
	\begin{tabular}{lll}
	\toprule
	\multicolumn{1}{c}{Source} & \multicolumn{1}{c}{Sufficient Condition}  & \multicolumn{1}{c}{Notes} \\
	\midrule
	\textbf{Classic Results} & &  \\
    \citet{Wald1951} & WARP for market demand & See \citet{Kehoe1998} \\
    \citet{Wald1951} & Gross substitutes & See \citet{ArrowHahn1991}  \\
    \citet{Gorman1953} & Identical homothetic utility & See \citet{Kehoe1998} \\
    \citet{Eisenberg1961} & Homothetic utility, collinear endowments & See also \citet{Chipman1974} \\
    \citet{MitjushinPolterovich1978} & $\RRA<4$, collinear endowments &     \\
    \citet{BalaskoShell1981_3} & Cobb-Douglas & OLG model \\
    \citet{KehoeLevineMas-CollelWoodford1991} & Gross substitutes & OLG model \\
    \citet{HensLoeffler1995}  & $\RRA \le 1$, separable & See also \citet{Slutsky1915}   \\
    \citet{HensLoeffler1995}  & CARA & See also \citet{Rubinstein1974} \\
    \citet{HensLoeffler1995}  & Common beliefs, constant supply, separable &    \\
    & Quadratic utility, common beliefs & Source unknown \\
	\midrule
    
    \textbf{Recent GE} & &  \\
    \citet{Chipman2010}  & CRRA, symmetry, $2$ goods/agents, $\RRA \le 2$ & Common Bernoulli utility  \\
    \citet{LoiMatta2018} & Eq. manifold zero curvature, $2$ goods & $e_i,x_i \in \R^J$  \\
    \citet{GeanakoplosWalsh2018Uniqueness}  & DARA, common endowments, $2$ goods & Common Bernoulli utility  \\
    \citet{GeanakoplosWalsh2018Uniqueness}  & CRRA, ordered endowments, $2$ goods & Common Bernoulli utility  \\
	\citet{Gimenez2022}  & Offer curve monotonicity for one good & $2$ goods/agents \\
	\citet{Won2023}  & CRRA, endowment growth bounds \eqref{eq:Won} & $2$ goods (can extend to $J$)  \\
    \citet{LoiMattaUccheddu2023} & Eq. manifold zero curvature, $2$ agents & $e_i,x_i \in \R^J$  \\
    \citet{LoiMatta2023} &HARA \eqref{eq:hara}, $\RRA \in \left(1,\frac{I}{I-1} \right]$, $2$ goods & $I$ agents  \\
   	\midrule
  
    \textbf{Macro} & &  \\
	\citet{Calvet2001}  &  CARA-normal (i.i.d.), no portfolio constraints  & $T<\infty$ GEI, no agg. risk  \\
	\citet{BloiseCitanna2015}  &  $\RRA \le 1$, $T=\infty$ limited commitment model  & \citet{GottardiKubler2015}   \\
	\citet{Toda2017JEDC}  & CARA, AR(1)/PT income, no credit limit & \citet{Huggett1993} model  \\
    \citet{Toda2019JME}  & CRRA, $\RRA \le 1$, no credit limit, $\beta$ risk only&  Life/death \citet{Huggett1993}  \\
	\citet{Light2020}  & CRRA, $\RRA \le 1$, firm $\ES \ge 1$ & \citet{Aiyagari1994} model  \\
    \citet{AchdouHanLasryLionsMoll2022}  & $\RRA \le 1$, no borrowing & Continuous \citet{Huggett1993}  \\
    \citet{AchdouHanLasryLionsMoll2022}  & $\RRA \le 1$, $\ES>$ capital share, no borrowing & Continuous \citet{Aiyagari1994}  \\
    \citet{Light2023}  & CES utility ($\mathrm{ES} \ge 1$), i.i.d. income  & Multi-good \citet{Huggett1993}  \\

	\midrule
    \textbf{Trade/Production} & &  \\
	\citet{Samuelson1951}  & CRS prod. network, $1$ factor, $1$ firm per good  & See also \citet{Stiglitz1970}  \\
    \citet{ScarfWilson2005}  & CRS prod., gross subst., no intermediates  & Many factor Ricardo  \\
    \citet{AlvarezLucas2007}  & Tradeable labor share $\le$ nontradeable share & \citet{eaton-kortum2002} \\
    \citet{AcemogluAzar2020}  & Generalization of \citet{Samuelson1951}  & Frictions/endog. network  \\
    \citet{AllenArkolakisTakahashi2020} & Demand elasticity $\le 0 \le $ supply elasticity & Universal gravity model  \\
   	\citet{AllenArkolakisLi2023} & Elasticity bound spectral radius $\le 1$ & General spatial model  \\
    \citet{KucheryavyyLynRodriguezClare2023a} & $\mathrm{EES}\times (\ES-1)\le 1$, $2$ countries or no frictions & EES Armington, no interm.  \\
	
	\bottomrule
	\end{tabular}%
}
\footnotesize
\raggedright
Note: The ``Sufficient Condition'' column gives our interpretation of each paper's key conditions for uniqueness of competitive equilibrium. However, as it is obviously impossible to fit full model details in a single table, we have omitted a variety of more mild or subtle assumptions for many of the papers. See the main text or the particular papers for further detail.
	\label{tab:papers}%
\end{table}%

\clearpage
\newpage

\printbibliography

\end{document}